\documentclass[12pt,english]{amsart}
\usepackage[inner=1in,outer=1in,top=1in,bottom=1in, marginparwidth=.7in,marginparsep=.1in]{geometry}
\usepackage{changes}
\usepackage{mathtools}
\usepackage{graphicx,multirow}
\usepackage[utf8x]{inputenc} 
\usepackage{enumerate,caption}
\usepackage{enumitem}
\usepackage{subcaption}

\numberwithin{equation}{section}
\usepackage{amssymb,latexsym,amsmath,amsthm,mathrsfs}
\theoremstyle{plain}
\newtheorem{theorem}{Theorem}[section]

\newtheorem{lemma}[theorem]{Lemma}

\newtheorem{proposition}[theorem]{Proposition}

\newtheorem*{thm*}{Theorem}

\theoremstyle{definition}

\newtheorem{remark}[theorem]{Remark}

\newtheorem{definition}[theorem]{Definition}

\newcommand\I{\iota}

\newcommand\C{\mathbb{C}}

\newcommand\T{\mathbb{T}}

\newcommand\R{\mathbb{R}}
\newcommand\Z{\mathbb{Z}}

\DeclareMathOperator\diag{diag}

\newcommand{\rev}[1]{{\color{black}{#1}}}

\title[Signal recovery from blind phaseless periodic short-time Fourier transform]{Near-optimal bounds for signal recovery from blind phaseless periodic short-time Fourier transform}
\author{Tamir Bendory}
\address{School of Electrical Engineering, Tel Aviv University}
\author{Chi-yu Cheng}
\address{Department of Mathematics, University of Missouri-Columbia}
\author{ Dan Edidin}
\email{bendory@tauex.tau.ac.il,  ccp9f@missouri.edu, edidind@missouri.edu}
\date{}
\begin{document}

\maketitle
\begin{abstract}
	We study the problem of recovering a signal $x\in\C^N$ from samples of its phaseless periodic short-time Fourier transform (STFT):  the magnitude of the  Fourier transform of the signal multiplied by a sliding window $w\in \C^W$. 
	We show that if the window $w$ is known, then a generic signal can be recovered, up to a global phase, from less than $4N$ phaseless STFT measurements.
	In the blind case, when the window is unknown, we show that the signal and the window can be determined simultaneously, up to a group of unavoidable ambiguities, from less than $4N+2W$ measurements. 
	In both cases, our bounds are optimal, up to a constant smaller than two.
\end{abstract}
%\rev{This is a text}
\section{Introduction}
The short-time Fourier transform (STFT) of a signal $x\in\C^N$ can be interpreted as the Fourier transform of the signal multiplied by a sliding window $w\in \C^W$ 
\begin{equation} \label{eq.stft}
	Y_{m,r}(x,w) = \sum_{n = 0}^{N-1} x[n] w[rL -n]e^{-2\pi\I nm/N},
\end{equation}
for  $0\leq m\leq N-1$ and $0\leq r\leq R-1$, 
where $L$ is the separation between sections, 
$R = N/\gcd(N,L)$ is the number of short time sections, and $w[n] =0$ for $W \leq n \leq N-1$.
We assume that all signals are periodic, and thus  all indices should be considered modulo~$N$.

This paper studies the fundamental conditions allowing  unique signal recovery---up to unavoidable ambiguities that will be  precisely defined later---from the magnitude of its STFT  $|Y_{m,r}(x,w)|$, namely, from its phaseless STFT measurements. 
In particular, we study two cases: 1) the window function $w$ is known, and  2) the blind case when $w$ is unknown and needs to be recovered simultaneously with the signal $x$.
We prove near-optimal bounds for both cases. For the known-window case, we show that no more than $4 N$ measurements suffice to recover the $2N$ parameters of $x\in\C^N$, substantially improving upon previous results~\cite{pfander2019robust,bojarovska2016phase}.
In the blind case, we prove that merely $\sim 4 N + 2W$ measurements determine the $2N+2W$ parameters that define the signal and window.  As far as we know, this is the first uniqueness result for the blind setup.  

Section~\ref{sec:main_results} introduces and discusses the main results of this paper, which are proved in Section~\ref{sec:proofs}.
It should be emphasized that our results concern only the question of uniqueness, and do not imply that practical algorithms can robustly recover the signal with only $O(N)$ measurements; the computational and stability properties of different algorithms were studied in~\cite{eldar2014sparse,iwen2016fast,bojarovska2016phase,bendory2017non,pfander2019robust,iwen2020phase,preskitt2021admissible,alaifari2021stability}. 
Nevertheless, in Section~\ref{sec:numerical_experiments}, we show numerical results suggesting that $O(N)$ might suffice for signal recovery, when the window is known. 

\paragraph{\textbf{Motivation.}} The motivation of this paper is twofold. First, phaseless STFT measurements naturally arise in ptychography:  a computational method of microscopic imaging, in which the specimen is scanned
by a localized  beam and Fourier magnitudes of overlapping  windows are
recorded~\cite{rodenburg2008ptychography,dierolf2010ptychographic,maiden2011superresolution,zheng2013wide,yeh2015experimental,pfeiffer2018x}.
The precise structure of the window might be unknown a priori and thus standard algorithms in the field optimize over the signal and the window
simultaneously~\cite{thibault2009probe,maiden2009improved,guizar2008phase}. 
This paper illustrates the fundamental conditions required for unique recovery in ptychography, regardless of the specific algorithm used. 
Second, this paper is part of  ongoing efforts to unveil the mathematical and algebraic properties standing at the heart of the phase retrieval problem---the problem of recovering a signal from phaseless measurements~\cite{shechtman2015phase,bendory2017fourier,grohs2020phase,barnett2020geometry,bendory2022algebraic}.  
Next, we succinctly present some of the main results in the field.  

\paragraph{\textbf{The phase retrieval problem.}} Phase retrieval is the problem of recovering a signal $x\in\C^N$ from   
\begin{equation} \label{eq:pr}
	y = |Ax|,
\end{equation}
for some sensing matrix $A\in\C^{M\times N}$, where the absolute value should be understood entry-wise.
In some cases, we may also assume prior knowledge on the signal, such as sparsity or known support. 
The phaseless periodic STFT setup  is a special case of~\eqref{eq:pr}, where the matrix $A$ represents samples of the STFT operator. 
The first mathematical and statistical works on phase retrieval focused on a random ``generic'' matrix $A$, see for example~\cite{candes2015phase,balan2006signal,balan2009painless,candes2015phasewirtinger,sun2018geometric,goldstein2018phasemax}.
These works were extended to the coded diffraction model~\cite{candes2015phase,gross2017improved}, which resembles our model, but the deterministic sliding window is replaced by a set of random masks. 
Unfortunately, the measurements in practice are not random, and thus this line of
work is of theoretical rather than applicable interest.
%remains of rather academic interest. 

In recent years, there has been a growing interest in deterministic phase retrieval setups  that better describe 
imaging applications. In particular, the \emph{non-periodic} phaseless STFT problem with a known window was studied in \cite{nawab1983signal,jaganathan2016stft,marchesini2016alternating,bendory2019blind}.
This setup differs from our case since out of range indices are set equal to zero, and there are $\left\lceil (N+W-1)/L \right\rceil$ distinct short-time sections instead of  
$R=N/\gcd(N,L)$ in the periodic case. 
The authors of~\cite{jaganathan2016stft} proved unique recovery with $\sim N$ samples, and also proposed a convex program to recover the signal. 
The blind case was studied  by two of the authors in~\cite{bendory2019blind}, who proved that the signal and the window
can be recovered, up to trivial ambiguities of dimension $L$, from $\sim 10(N+W)$ 
measurements.  In this work, we show that in the periodic case, $\sim 4N$ and  $\sim 4N+2W$ measurements are enough in the known-window case and blind case, respectively. The continuous STFT setup was studied in~\cite{alaifari2021gabor,grohs2019stable,grohs2021stable,fannjiang2020blind}.

More phase retrieval applications whose fundamental conditions for unique recovery were studied include ultra-short pulse characterization using frequency-resolved optical gating (FROG)~\cite{bendory2020signal,bendory2017uniqueness,trebino2000frequency} or using multi-mode fibers~\cite{bendory2021signal,xiong2020deep,ziv2020deep}, X-ray crystallography (recovering a sparse signal from its Fourier magnitude)~\cite{elser2018benchmark,bendory2020toward}, recovering a one-dimensional signal from its Fourier magnitude~\cite{huang2016phase,MR3842644,beinert2015ambiguities,edidin2019geometry}, holographic phase retrieval~\cite{barmherzig2019holographic,barmherzig2019dual}, and vectorial phase retrieval~\cite{raz2013vectorial,raz2011vectorial}.

%The next section introduces and discusses the main results of this paper, and 
%Section~\ref{sec:motivation} elaborate about  its motivation.

\section{Main results} \label{sec:main_results}
We begin by stating our result for the known-window case.
\begin{theorem}[Known window]\label{thm.known_window}
	For a generic known window vector $w \in \C^W$, a generic vector $x \in \C^N$ can be recovered, up to a global phase, from
	\begin{equation*}
		2(2W-1) + \left\lceil{{(4 \alpha -1)(N - (W+\alpha))}\over{\alpha}}\right\rceil
%\lceil 4N - \frac{N-W}{\alpha}-4\alpha-1 \rceil		
	\end{equation*}
	phaseless periodic STFT measurements of step length $L$, where $\alpha =
	\gcd(L,N)$. 
\end{theorem}

\begin{remark}
We say that a condition holds for generic signals $x$ and windows $w$ if the set of signals and windows for which the condition does not hold is defined by polynomial conditions. In particular, the set of pairs $(x,w) \in \C^N \times \C^W$
for which the conclusion of Theorem \ref{thm.known_window}  {holds is dense and its complement} has measure zero. {For a precise definition of the term generic see Definition \ref{def.generic}.}
\end{remark}

It is not hard to deduce that Theorem~\ref{thm.known_window} implies that the number of required measurements for signal recovery is smaller than $$4N-\frac{N-W}{\alpha}-2<4N,$$ while the number of parameters to be recovered is $2N$.
%We note that the bound increases linearly with $N$ and $W$ and non-linearly with $\alpha$. 
%In particular,
If $N$ is a prime number, then $\alpha =1$ (independently of $L$) and the bound improves to $\sim 3N+W$. For a long window $W\approx N$, the bound tends to $4N$. 
Figure~\ref{fig:known_window} presents the bound of Theorem~\ref{thm.known_window} for $N = 100$ as a function of~$W$ for various values of $L$. 
As can be seen, the curves are bounded by $4N$. 
%If $N$ is a prime number, then $\alpha=1$ and the curves of all values of $L$ coincide. 

\begin{figure}
	\begin{subfigure}[ht]{0.45\columnwidth}
		\centering
		\includegraphics[width=\columnwidth]{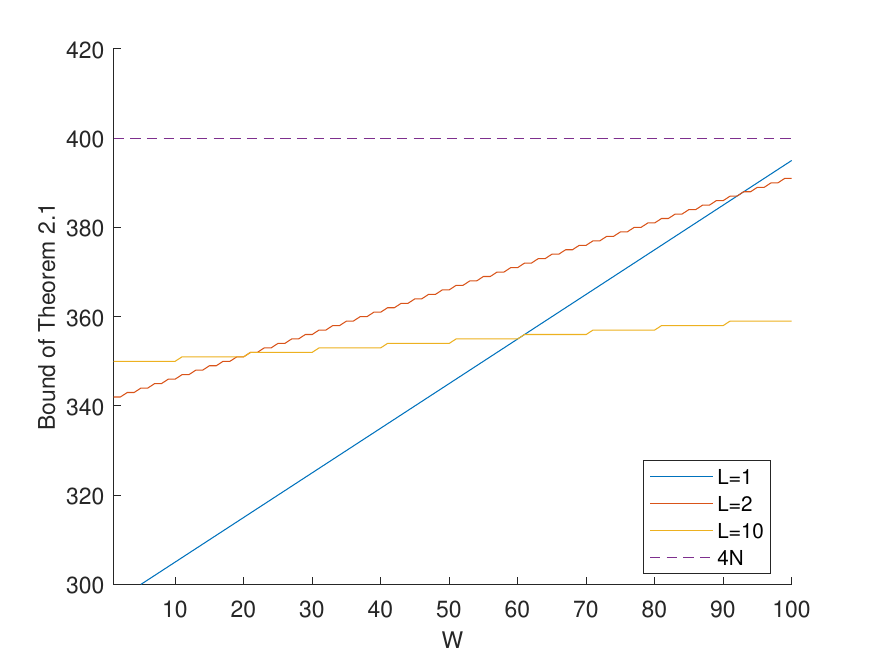}
		\caption{\label{fig:known_window} The bound of Theorem~\ref{thm.known_window}}
	\end{subfigure}
	\hfill
	\begin{subfigure}[ht]{0.45\columnwidth}
		\centering
		\includegraphics[width=\columnwidth]{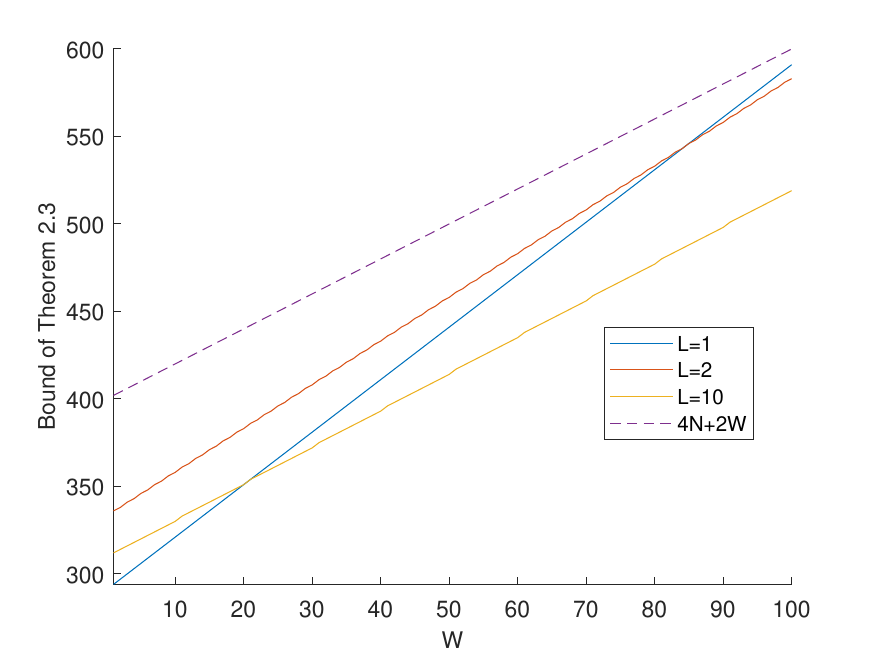}
		\caption{\label{fig:blind} The bound of Theorem~\ref{thm.unknown_window}}
	\end{subfigure}
	\caption{\label{fig:theoretical_bounds} The bounds of Theorems~\ref{thm.known_window} and~\ref{thm.unknown_window} for $N=100$  as a function of the window length $W$, for various values of $L$.}
\end{figure}

\rev{
 \begin{remark}
 Given a vector $y \in \C^N$, let $T_\ell y$ denote the cyclically shifted
  vector defined by $(T_\ell y)[n] = y[n - \ell]$ with all indices taken modulo
  $N$. Likewise, define the modulated vector $M_my$ by setting $(M_my)[n] = \omega^{mn} y[n]$, where $\omega = e^{2\pi \iota/N}$.
  For a given generic window vector $w \in \C^W$, the 
    vectors $f_{m,r} = M_mT_{rL}w$ form an $NR$-element frame in $\C^N$ consisting of
  vectors whose supports all have length $W$.
    With this notation, the phaseless STFT measurement $|Y_{m,r}(x)|$
   equals to the phaseless frame measurement $|\langle x,f_{m,r}\rangle |$.
    Theorem \ref{thm.known_window} implies that a subset of the $\{f_{m,r}\}$ forms a highly structured frame with less than $4N$ elements for which it is possible to recover a generic vector,  up to global phase, from its phaseless frame measurements.
     By contrast, \cite[Theorem 3.4]{balan2006signal} implies that if $M \geq 2N$ then for a generic $M$-element frame it is possible to recover a generic vector, up to a global phase, from its phaseless frame measurements. Also, note that if $M \geq 4N-4$ then \cite[Theorem 1.1]{conca2015algebraic} states that for a generic $M$-element frame {\em every} vector can be recovered, up to  a global phase, from its phaseless frame measurements.
  \end{remark}
}
  
Our second result deals with the blind case where the window $w$ is unknown, and therefore there are $2N+2W$  parameters to be recovered. 
%In this case, the dimension of the ambiguity group is $\alpha+2$---substantially larger than the dimension one ambiguity in the 
%known window case---and is presented explicitly in Proposition~\ref{prop.Gaction}.

\begin{theorem}[Unknown window] \label{thm.unknown_window}
	A generic pair $(x, w) \in \C^N \times \C^W$ can be recovered, up to
	a group   of trivial ambiguities of dimension $\alpha+2$ defined in  Proposition~\ref{prop.Gaction}, from at most
	\begin{equation*}
%3(2W-1) + (4\alpha -1){N- (W+\alpha)\over{\alpha}} 
%4N+2W-\frac{N-W}{\alpha} - 2-4\alpha%,
3(2W-1) + \left\lceil {{(4 \alpha-1)(N -(W +2\alpha))}\over{\alpha}} \right\rceil
	\end{equation*}
	phaseless periodic STFT measurements of step length $L$, where $\alpha = \gcd(L,N)$. 
\end{theorem}
Once again, the set of pairs $(x,w) \in \C^N \times \C^W$
for which the conclusion of Theorem \ref{thm.known_window} {holds is dense and its complement} has measure 0.

Theorem~\ref{thm.unknown_window} shows that the number of measurements is bounded by 
$$4N+2W-\frac{N-W}{\alpha} - 3<4N +2W,$$
exceeding the number of parameters to be recovered by a constant smaller than 2. 
%As in the non-blind case, the bound increases linearly in $N$ and $W$. %and non-linearly in $\alpha$. 
For $\alpha=1$, the bound reads $\sim 3N+3W$: much smaller than $4N+2W$ for $W\ll N$, which is the typical situation in ptychography---a chief motivation of this paper.
However, in contrast to the known-window case, in the blind case $\alpha$ has a big impact on the dimensionality of the ambiguity group:  the dimension of the ambiguity group is $\alpha+2$, substantially larger than the dimension one ambiguity in the 
known-window case.
 Therefore, if possible, in this case it is preferable to choose a prime $N$.  
Figure~\ref{fig:blind} presents the bound of Theorem~\ref{thm.unknown_window} for $N = 100$ as a function of $W$. 

%\begin{figure}
%	\begin{subfigure}[ht]{0.45\columnwidth}
%		\centering
%		\includegraphics[width=\columnwidth]{theoretical_bounds_blind_N100.pdf}
%		\caption{$N=100$}
%	\end{subfigure}
%	\hfill
%	\begin{subfigure}[ht]{0.45\columnwidth}
%		\centering
%		\includegraphics[width=\columnwidth]{theoretical_bounds_blind_N101.pdf}
%		\caption{$N=101$}
%	\end{subfigure}
%	\caption{\label{fig:theoretical_bounds_blind} The bound of Theorem~\ref{thm.known_window} for $N=100$ and $N=101$ (a prime number) as a function of the window length $W$. The red dashed line indicates $4N+2W$.}
%\end{figure}

The proofs of both theorems rest on extensions of technical results proved in~\cite{bendory2019blind}. The key point is that $\sim 4W$ (known window) or $\sim 6W$ (unknown window) phaseless periodic STFT measurements determine the Fourier intensity functions of short sequences
of vectors in $\C^W$ that satisfy certain polynomial constraints. Using the method of \cite[Theorem 5.3]{edidin2019geometry}, we show that the Fourier phase retrieval problem is solvable for generic vectors satisfying these constraints. Knowledge of these short sequences gives information about some of the entries in the vector $x$
and in the blind case fully determine the window.
We then use~\cite[Proposition IV.2]{bendory2019blind} to bound the number of further phaseless STFT measurements needed to fully determine the signal $x$.

\section{Proofs} \label{sec:proofs}

\subsection{Preliminaries} 
\subsubsection{Notation about the discrete Fourier transform} \label{sec.notation_fourier}
In this section we establish some notation about the discrete Fourier transform and Fourier intensity function. For a reference, see \cite{beinert2015ambiguities, edidin2019geometry}.

If $y \in \C^W$ is a vector, let 
$\hat{y}(\omega)=y[0] + y[1] \omega + \ldots y[W-1] \omega^{W-1}$ be the polynomial on the unit circle $\omega = e^{-\iota \theta}\in S^1$. %and $\omega = e^{-\iota \theta}$
%is a coordinate on the circle. 
The discrete Fourier transform vector
$\hat{y}$ is obtained by evaluating this polynomial at the $W$-th roots of unity;
i.e.,
$$\hat{y} = \left(\hat{y}(1), \hat{y}(\eta), \ldots, \hat{y}(\eta^{W-1})\right),$$
where $\eta = e^{-2\pi \iota/W}$.

By abuse of notation, we will sometimes view
$\omega$ as a coordinate on the entire complex plane and then we can speak
about the {\em roots} of $\hat{y}(\omega)$. We typically assume that our vectors
satisfy $y[0], y[W-1] \neq 0$ so the polynomial $\hat{y}(\omega)$ will have
$W-1$ (not necessarily distinct) roots in~$\C$.
If $(\beta_1, \ldots, \beta_{W-1})$ are the roots of $\hat{y}(\omega)$, 
then we can write
$$\hat{y}(\omega) = y[W-1](\omega - \beta_1) \ldots (\omega - \beta_{W-1}).$$

Given a vector $y=(y[0], \ldots , y[W-1])$, the Fourier intensity  of
$y$ is $A_y(\omega) = |\hat{y}(\omega)|^2$. Expanding out and using the fact that
$\overline{\omega} = \omega^{-1}$ on the circle $S^1$, the Fourier
intensity function factors as \cite{edidin2019geometry}
\begin{equation}
	A_y(\omega) = \omega^{1-W} \overline{y[0]} y[W-1]
	(\omega - \beta_1)\left(\omega - {1\over{\overline{\beta_1}}}\right)
	\ldots (\omega - \beta_{W-1})\left(\omega - {1\over{\overline{\beta_{W-1}}}}\right).
\end{equation}
(Note that for any complex number $\beta$,
${1\over{\overline{\beta}}}= {\beta\over{|\beta|^2}}$, a fact we will use extensively.)
If $y'$ is another vector
such that $A_y = A_{y'}$, then the proof of \cite[Theorem 3.1]{MR3842644} implies
$$\hat{y'}(\omega) = e^{\iota \theta}|y[W-1]| \prod_{i \in I}|\beta_i|\left(\omega - {1\over{\overline{\beta}_i}}\right) \prod_{i \notin I} (\omega - \beta_i),$$
for some subset $I \subset [1, W-1]$.

%The STFT measurements are the $NR$ complex numbers
%\begin{equation} \label{eq.ymr}
%	Y_{m,r}(x,w) = \sum_{n = 0}^{N-1} \eta_m^{n} x[n] w[rL -n],\text{ for }0\leq m\leq N-1\text{ and }0\leq r\leq R-1
%\end{equation}
%where $\eta_m = e^{-2\pi \iota m/N}$ (so $\eta_m = \eta_1^m$). Because our signal
%is assumed to be peridoic the indices are taken modulo $N$ and $w_n =0$ for $W \leq n \leq N-1$. Note that this differs from the non-periodic case
%where out of range indices are set equal to 0. It also means that there are
%$R=N/\gcd(N,L)$ distinct short-time sections instead of $\lceil N+W-1/L \rceil$ in the non-periodic case.
\subsubsection{Notation for the STFT measurements} 
For our proofs, it is convenient to use the fact that $x$ is periodic and that
$w[n] =0$ for $W \leq n \leq N-1$ to rewrite the STFT~\eqref{eq.stft} as
\begin{equation} \label{eq.ymr_countdown}
	Y_{m,r}(x,w) = \eta_m^{rL}\sum_{n=0}^{N-1} \eta_{-m}^n x[rL -n]w[n],
\end{equation}
where $\eta_m:=e^{2\pi \I m/N}$, so $\eta_{-m}:=e^{-2\pi \I m/N}$ and $\eta_m^n:=e^{2\pi \I mn/N}$.
Let  $T_{rL} x \circ w$, where $$T_{rL}x = (x[rL], x[rL -1], \ldots, x[N-1-rL])\in\mathbb{C}^{N}$$
be the vector $x$ shifted by $rL$, and $\circ$ denotes the entry-wise product.
Thus, for fixed~$r$, the measurements $\{Y_{m,r}\}_{m=0}^{N-1}$
determine $N$ values of the Fourier
transform of the vector
$y_{rL} = T_{rL} x\circ w$,
where the indices are taken modulo $R$.
The phaseless STFT measurements $|Y_{m,r}|_{m=0}^{N-1}$ give $N$
values of the Fourier intensity function $A_{y_{rL}}$ of the vector
$y_{rL}$.

\rev{
  \subsubsection{Terminology from algebraic geometry}
\begin{definition} \label{def.generic}
   A property ${\bf P}$ holds {\em generically} on $\C^M$ if the set
  $Z\subset \C^M$ where property ${\bf P}$  does not hold is contained in a subset $Y$ of $\C^M$ defined by a non-zero polynomial. More generally, if
  $X \subset \C^M$ is a subset defined by polynomial equations, then a property
  ${\bf P}$ holds generically on
  $X$ if the set $Z \subset X$ where property ${\bf P}$ does not hold is contained in a subset of $X$, which is defined by a polynomial which does not vanish identically on $X$.
\end{definition}
}

\subsection{Proof of Theorem~\ref{thm.known_window}}
Since $w$ and $x$ are generic,  we  assume that $w[0], \ldots , w[W-1]$
and $x[0], \ldots , x[N-1]$ are all non-zero. By applying the action
the group of ambiguities, $S^1$, we can also assume that $x[0]w[0]$ is real and positive.
Since $w$ is fixed and known, in this section we will use the notation
$Y_{m,r}(x)$ instead of $Y_{m,r}(x,w)$.

Let $x'$ be a solution to the system of quadratic equations
$\{|Y_{m,r}(x')|^2 = |Y_{m,r}(x)|^2\}$. We will use 
a recursive method
to show that for generic $x$, there is a unique solution $x'$ 
with $x'[0]w[0]$ real and positive and that $x'$ can be determined
using at most
%\TODO{We don't have the extra $\alpha$ in Theorem~\ref{thm.known_window}.}
$$2(2W-1) + \left\lceil (4 \alpha -1) {N - (W+\alpha)\over{\alpha}}\right\rceil < 4N$$ phaseless
STFT measurements.
The proof consists of two main stages, outlined below. 

\subsubsection{Step 1. Determining $x[\alpha], x[\alpha-1], \ldots , x[-W+1]$ with
$4W-2$ phaseless STFT measurements}
Using $2W-1$ phaseless measurements of the form $|Y_{m,0}|$ for $2W-1$ different values of $m$  we can obtain the Fourier intensity function of the vector
$$y_0 = T_0x \circ w = (x[0]w[0], x[-1]w[1],\ldots , x[-W+1]w[W-1]).$$
Likewise,  $2W-1$ phaseless measurements of the form $|Y_{m,r_1}|$, 
where $r_1L \equiv \alpha \bmod R$, determine
the Fourier intensity function of the vector
$$y_\alpha = T_\alpha x \circ w = (x[\alpha]w[0], \ldots , x[\alpha -W +1]w[W-1]).$$ Note that because $\alpha = \gcd(L,N)$ and $R = N/\alpha$, there
is a unique $r_1$ with $0 < r_1 \leq R-1$ such that $r_1L \equiv \alpha \bmod R$.
The two vectors $y_0$ and $y_\alpha$ are not algebraically independent as they satisfy the linear equations

 \begin{equation}\label{eq.linear_relations}
    w[j+\alpha] y_{0}[j] = w[j]y_{\alpha}[j+\alpha], \quad  j=0, \ldots, W-1-\alpha.
 \end{equation} %\tb{$y_{\alpha}[j+\alpha]$?}
 The proof of the following result is somewhat technical and is given in Appendix~\ref{sec:proof_prop_linhay}.
 Recall from Section \ref{sec.notation_fourier}
 that if $y \in \C^W$, $A_y$ denotes the Fourier intensity function
 $|\hat{y}(\omega)|^2$.
 \begin{proposition} \label{prop.linearhay}
   A generic pair of vectors $(y_0, y_\alpha)$
   %\tb{$(y_0, y_\alpha)$?}
   satisfying equations \eqref{eq.linear_relations} is determined, up to a global phase, from the Fourier intensity functions of $y_0$ and $y_\alpha$.
   Precisely, if $(y'_0, y'_\alpha)$ is a pair of vectors
   satisfying equations~\eqref{eq.linear_relations} such that $A_{y_0} = A_{y'_0}$
   and $A_{y_\alpha} = A_{y'_\alpha}$,   then $(y'_0,y'_\alpha) = e^{\iota \theta}(y_0,y_\alpha)$ for some
   $e^{\iota \theta} \in S^1$.
 \end{proposition}
  We also need the following lemma. 
 \begin{lemma} \label{lem.onto}
   If all coordinates of $w$ are non-zero, then for any pair
   $(y_0, y_\alpha)$ 
   satisfying equations~\eqref{eq.linear_relations}
   there exists a vector $x$ such that $(y_0, y_\alpha) =
   (x \circ w, T_\alpha x \circ w)$.
 \end{lemma}
  \begin{proof}
    Given $y_0, y_\alpha$ satisfying~\eqref{eq.linear_relations}, 
    define a vector $x$ by setting 
    \begin{equation*}
    x[n] =	\begin{cases}
    	y_0[n]/w[-n]& \quad \text{if} \quad -W +1 \leq n \leq 0, \\
     	y_\alpha[n]/w[\alpha -n]& \quad  \text{if} \quad 0 <  n \leq \alpha, \\
     	\text{arbitrary}& \quad \text{else}.
    	\end{cases}
    \end{equation*}    
%    $x'[n] = y_0[n]/w[-n]$ if $-W +1 \leq n \leq 0$ and $x'[n] = y_1[n]/w[\alpha -n]$ if $0 <  n \leq \alpha$ and all other entries of $x'$ can be arbitrary. 
    Then, it is easy to check that
    $(y_0,y_\alpha) = (x \circ w, T_\alpha x \circ w)$.
  \end{proof}
  Proposition \ref{prop.linearhay} and Lemma \ref{lem.onto} imply that for
  generic $(x,w)$ the vectors $x \circ w$ and $T_\alpha \circ w$ are uniquely
  determined, up to a global phase, by $2(2W-1)$ phaseless STFT measurements
  of the form $|Y_{0,m}(x)|$ and $|Y_{r_1,m}(x)|$. In particular, if $x'$
  is another vector such that $|Y_{m,0}(x')| = |Y_{m,r_1}(x')|$ for $2W-1$
  distinct values of $m$, then $(x' \circ w, T_\alpha x' \circ w) = e^{\iota \theta}(x\circ w, T_\alpha x \circ w)$. By imposing the condition that $x[0]w[0]$
  is positive real, we can eliminate the global phase ambiguity
  and conclude that $(x'\circ w, T_\alpha x' \circ w) = (x\circ w, T_\alpha x \circ w)$.
  In other words,
  $(x'[0]w[0], \ldots x'[-W +1]w[W-1]) = (x[0]w[0], \ldots , x[W-1]w[W-1])$ and $(x'[\alpha]w[0], \ldots x'[\alpha -W +1] w[W-1]) = (x[\alpha]w[0],\ldots x[\alpha -W+1]w[W-1])$. If we assume that $w[0], \ldots , w[W-1]$ are non-zero, then
  it follows that $x'[n] = x[n]$ for $-W+1 \leq n \leq \alpha$.
  Therefore, we conclude that the $2(2W-1)$ phaseless STFT measurements determine
  $W+\alpha$  entries of the signal $x$, namely, $x[-W+1], x[-W+2],
  \ldots x[0], \ldots x[\alpha]$.

  \subsubsection{Determining the remaining $N-(W+\alpha)$
    entries of $x$ using $(4\alpha -1)\left\lceil \frac{N-(W+\alpha)}{\alpha} \right\rceil$ phaseless STFT measurements.}
  Consider the vector
  $$y_{2\alpha} = (x[2\alpha]w[0], \ldots , x[\alpha+1]w[\alpha -1], x[\alpha]w[\alpha], \ldots , x[-W + 2\alpha +1]x[W-1]).$$
  By Step 1 we know the entries $y_{2\alpha}[n]$ for $n \in [\alpha, W-1] \subset
  [0, W-1]$. In particular, all unknown entries of $y_{2\alpha}$
  lie in the subset $S = [0, \alpha-1]$ of $[0, W-1]$.
  Hence, by~\cite[Proposition IV.3, Corollary IV.4]{bendory2019blind}, 
   a generic vector $y_{2\alpha}$ can
  be recovered from the values of its Fourier
  intensity function $A_{y_{2\alpha}}$ at 
  $2|S-S| -1 + 2|S|$ distinct roots of unity.
  In our case,  $|S| = |S-S| = \alpha$.
  Hence, $y_{2\alpha}$ can be recovered from
  the value of $A_{y_{2\alpha}}$ at  $4 \alpha -1$ distinct roots of unity. Now, the phaseless
  STFT measurements $|Y_{m,r_2}|$, where $r_2 L \equiv 2\alpha \bmod N$, 
  are the values of the Fourier intensity function
  of $y_{2\alpha}$. Hence, we can recover $y_{2\alpha}$ from $|Y_{m,r_2}|$ for $4\alpha -1$ values of $m$.

  We can now complete the proof by induction. If $x[-W+1], \ldots , x[j \alpha]$
  are known, then we require $4 \alpha-1$ phaseless measurements of
  the form $|Y_{r_j,m}|$ to determine the next $\alpha$ entries
   $x[j\alpha +1], \ldots x[(j+1)\alpha]$ of $x$. 
(Here, $r_{j}L \equiv j\alpha \bmod R$).
    It follows that we can determine all entries of $x$ from (at most)
  $$2(2W-1) + \left\lceil \frac{(4\alpha -1)(N-(W+\alpha))}{\alpha} \right\rceil$$ phaseless STFT measurements.

  \subsection{Proof	of Theorem~\ref{thm.unknown_window}}

In this section, we prove that even if the window is not known, we can recover a
generic pair $(x,w) \subset \C^N \times \C^W$ from $\sim 4N+2W$ measurements, up to the action
of the group $G$ of trivial ambiguities. The strategy of our proof follows the proof of Theorem~\ref{thm.known_window}. We begin by explicitly define the group of ambiguities.

\subsubsection{The group of ambiguities} \label{sec.blind_ambiguities}
%We begin by describing the group of ambiguities.

Let $G$ be the group $S^1 \times (\C^*)^\alpha \times \Z_R$,
where we identify $\Z_R$ with the group of $R$-th roots of unity.
We define an action of $G$ on $\C^N \times \C^W$ as follows:
\begin{itemize}
\item  $e^{\iota \theta}\in S^1$ acts by $e^{\iota \theta}(x,w) =
  (e^{\iota \theta}x, e^{\iota \theta}w)$.
  
\item $ \lambda = (\lambda[0], \ldots, \lambda[\alpha-1]) \in (\C^*)^\alpha$
  acts on $x$ 
  by 
   \begin{equation*}
  	\begin{split}
  		\left(\lambda[0] x[0], \lambda[\overline{1}]x[1], \ldots , \lambda[\overline{N-1}]
  		x[N-1]\right),
  	\end{split}	
  \end{equation*}
and on $w$ by
 \begin{equation*}
	\begin{split}
		\left(\lambda[0]^{-1} w[0], \lambda[\overline{-1}]^{-1} w[1], \ldots ,
		\lambda[\overline{-W+1}]^{-1}w[W-1]\right),
	\end{split}	
\end{equation*}
%  \begin{equation*}
%  	\begin{split}
%  	\lambda(x,w)&=
%  	\left((\lambda[0] x[0], \lambda[\overline{1}]x[1], \ldots , \lambda[\overline{N-1}]
%  	x[N-1]), \right.\\ & \left.  (\lambda[0]^{-1} w[0], \lambda[\overline{-1}]^{-1} w[1], \ldots ,
%  	\lambda[\overline{-W+1}^{-1}]w[W-1])\right),
%  	\end{split}	
%  \end{equation*}
  where $\overline{j}$ indicates the residue of $j$ modulo $\alpha$.
\item  If $\omega$ is an $R$-th root of unity, then $\omega$ acts by
  $\omega(x,w) = (x',w')$, where $x'[n] = \omega^{\lfloor n/\alpha \rfloor} x[n]$
  and $w'[n] = \omega^{\lceil n/\alpha \rceil} w[n]$.
  Note that since $R | N$ this action is well defined even though our indices
  are always taken
  modulo $N$.
  \end{itemize}
\begin{proposition} \label{prop.Gaction}
If $g \in G$ then for all $m,r$, we have $|Y_{m,r}(x,w)| =|Y_{m,r}(g(x,w))|$; i.e., the phaseless STFT periodic STFT measurements are invariant under the action of $G$.
\end{proposition}
\begin{proof}
  The action of $S^1$ on $\C^N \times \C^W$ clearly preserves the magnitude
  of the STFT measurements.
  The STFT measurements are measurements of Fourier transform of the
  vectors
  $y_{j\alpha}(x,w) = (x[j\alpha]w[0], \ldots , x[j\alpha -W +1]w[W-1])$,
  where $j\in [0,R-1]$ is defined by equation $j \alpha \equiv rL \bmod N$.
  If $\lambda = (\lambda_0, \ldots , \lambda_{\alpha -1})$,
  then $y_{j\alpha}(\lambda(x,w))[n] =
  \lambda[\overline{j\alpha}-n] \lambda[\overline{-n}]^{-1} x[j\alpha-n]w[n].$
  Since $j\alpha -n \equiv -n \bmod \alpha$, we see that $y_{j\alpha}(\lambda(x,w))[n] = y_{r\alpha}(x,w)[n]$. In other words, the action of $(\C^*)^\alpha$ preserves
  the STFT measurements.
  Finally, if $\omega^{R} = 1$ then
  $y_{j\alpha}(\omega(x,w))[n] = \omega^{j}y_{j\alpha}(x,w)$.
  Hence, the $y_{j\alpha}(x,w)$ and $y_{j \alpha}(\omega(x,w))$ have the same Fourier intensity functions.
\end{proof}

\subsubsection{Strategy of the proof of Theorem \ref{thm.unknown_window}}
Our goal is to prove that for generic $(x,w)$, if $|Y_{m,r}(x',w')| = |Y_{m,r}(x,w)|$ then $(x',w')$ is related to $(x,w)$ by the action of the ambiguity group $G$. Moreover, we will show that we can determine $(x', w')$ using at most
$$3(2W-1) + \left\lceil {{(4 \alpha-1)(N -(W +2\alpha))}\over{\alpha}} \right\rceil$$
STFT measurements.

To begin, by applying the $S^1 \times (\C^*)^\alpha$ factor in $G$, we may assume
that $w[0], \ldots , w[\alpha -1]$ are known (for example, we can assume that they are all equal to 1) and that $x[0]$ is positive real. Hence, our goal is to show
that if $(x',w')$ is a solution for $|Y_{m,r}(x', w')| = |Y_{m,r}(x,w)|$ with
$x'[0]$ positive real and $w'[0] \ldots w'[\alpha-1] =1$, then $(x',w')$
is obtained from $(x,w)$ by the action of the group of $R$-th roots of
unity.

\subsubsection{Recovery of $y_0, y_\alpha, y_{-\alpha}$, up to a phase, from $3(2W-1)$ measurements.}
Consider the three vectors 
\begin{enumerate} 
    \item $y_{-\alpha}:=(x[-\alpha]w[0],\ldots,x[-\alpha-(W-1)]w[W-1])$;
    \item $y_{0}:=(x[0]w[0],\ldots,x[-(W-1)]w[W-1]])$;
    \item $y_{\alpha}:=(x[\alpha]w[0],\ldots, x[0]w[\alpha], \ldots
      x[\alpha -(W-1)]w[W-1])$.
\end{enumerate}
The phaseless $3(2W-1)$ measurements of the form $|Y_{m,0}(x,w)|,
|Y_{m,r_1}(x,w)|, |Y_{m,r_{-1}}(x,w)|$, for $2W-1$ distinct values of $m$,  determine the Fourier intensity functions $A_{y_0}, A_{y_\alpha}, A_{y_{-\alpha}}$, respectively. Here $r_1, r_{-1} \in [0,R-1]$ are defined by
the condition that $\alpha \equiv r_1 L \bmod N$ and $-\alpha \equiv r_{-1} L \bmod N$.

The triple $(y_{0},y_{-\alpha},y_{\alpha})$ satisfies the quadratic relations 
\begin{equation}\label{eq_blind_triple}
y_{-\alpha}[\ell]y_{\alpha}[\ell + \alpha] = y_{0}[\ell]y_{0}[\ell+\alpha], \quad \ell=0,\ldots,W-1- \alpha.
\end{equation}
By construction, the map $\Phi:\mathbb{C}^{N}\times\mathbb{C}^{W}\rightarrow
\mathbb{C}^{W}\times\mathbb{C}^{W}\times\mathbb{C}^{W}$, $\Phi(x,w)=
(y_{0},y_{-\alpha}, y_{\alpha})$,  
has image contained in the algebraic subset of $(\C^W)^3$
by equations~\eqref{eq_blind_triple}. Let $Z$ be the closure of the image.
The following proposition is proved in  Appendix~\ref{sec:proof_prop_hardhay}.
\begin{proposition} \label{prop.hardhay}
  For generic $(z_0,z_\alpha, z_{-\alpha}) \in Z \subset (\C^W)^3$,
  if $(z'_0, z'_\alpha, z'_{-\alpha}) \in Z$ have the same Fourier intensity
  functions as $(z_0,z_\alpha, z_{-\alpha})$,
  then there are angles $\theta_0, \theta_\alpha$ such that
  $z'_0 = e^{\iota \theta_0}z_0$, $z'_\alpha = e^{\iota (\theta_0 + \theta_{\alpha})}z_\alpha$, 
  and $z'_{-\alpha} = e^{\iota(\theta_0- \theta_\alpha)} z_{-\alpha}$.
\end{proposition}
%\end{document}
Applying the action of the subgroup $S^1 \times \C^*$
of the ambiguity group $G$ we may assume that $x[0]$ is real and positive and that 
$$w[0] =\ldots  = w[\alpha-1]=1.$$

%that $x[0]$ is positive real and $w[0] =1$ implies that we eliminate
%the phase indeterminacy of
%$$y_0(x,w) = (x[0]w[0], x[-1] w[1], \ldots, x[-W +1]w[W-1]).$$
It then follows from Proposition \ref{prop.hardhay} that 
if $(x',w')$ is a pair such that 
$|Y_{m,r}(x',w')| = |Y_{m,r}(x,w)|$ for $2W-1$ distinct values of $m$ for
 $r = 0, r_1, r_{-1},$
then we may assume
that 
\begin{eqnarray*}
y_0(x',w') & = &  y_0(x,w),\\
y_\alpha(x',w') & = &e^{\iota \theta_\alpha}y_\alpha(x,w),\\
y_{-\alpha}(x',w') &= & e^{-\iota\theta_\alpha}y_{-\alpha}(x,w).
\end{eqnarray*}
and 
%Since we have applied a trivial ambiguity to assume that $x'[0]$ and $x[0]$
%are positive real and that 
$$w'[0] = w[0],\ldots w'[\alpha-1] =w[\alpha -1].$$ 
%that the angle $\theta_0$ is 0; i.e., we may assume
%that 
It follows that $x'[-\ell]=x[-\ell]$ for $\ell=0,\ldots,\alpha-1$. The equality $y_{\alpha}(x',w')[\alpha+\ell] = e^{\iota \theta_\alpha}
y_{\alpha}(x,w)[\alpha+l]$ implies $$w'[\alpha+\ell]=e^{\iota\theta_{\alpha}}
w[\alpha+\ell].$$

Since $y_{0}(x',w')[\alpha+\ell]=y_{0}(x,w)[\alpha+\ell]$,
we conclude that $$x'[-\alpha-\ell]=e^{-\iota\theta_{\alpha}} x[-\alpha-\ell].$$
The equality $y_{\alpha}(x',w')[2\alpha+\ell]=e^{\iota \theta_\alpha}
y_{\alpha}(x,w)[2\alpha+\ell]$, then implies that
$$w'[2\alpha+\ell]=e^{2\iota\theta_{\alpha}}w[2\alpha+\ell].$$
Going back to $y_{0}(x',w')[2\alpha+\ell]$ and $y_{0}(x,w)[2\alpha+\ell]$, we deduce that 
$$x'[-2\alpha-\ell]=e^{-2\iota\theta_{\alpha}}x[-2\alpha-\ell].$$ This procedure goes on. In the end, we conclude that
\begin{equation*}
  w'[n]=e^{\iota \lfloor n/\alpha\rfloor \theta_{\alpha}}w[n],
\end{equation*}
for $n = 0, \ldots W-1$
and
\begin{equation*}
x'[m]= e^{\iota \lceil m/\alpha \rceil  \theta_{\alpha}} x[m],
\end{equation*}
for the $W+2\alpha$ values
$m = \alpha, \alpha-1, \ldots, 0, \ldots , -(W-1+\alpha)$.

\subsubsection{Determining the other values of $x[n]$}
We can now proceed recursively to compute $x[n]$ for $n \notin [-W+1-\alpha, \alpha]$.
Consider the vector $$y_{2\alpha}[x,w] = (x[2\alpha]w[0], x[2\alpha -1]w[1], \ldots x[\alpha]w[\alpha], \ldots , x[2\alpha -W +1]w[W-1]).$$
By our first step, we know the $W-\alpha$ entries of $y_{2\alpha}$ up to the
unknown common phase $e^{2\iota \theta_\alpha}$.
Precisely, $y_{2\alpha}(x',w')[n] = e^{2 \iota \theta_\alpha} y_{2\alpha}(x,w)[n]$
%\TODO{typo}
for $n \geq \alpha$. In particular, we know the last $W-\alpha$ entries
of the vector $z_{2\alpha} = y_{2\alpha}(x',w')/(x'[\alpha]w'[\alpha])$.
(Note that we assume that $x[\alpha]w[\alpha]$ is non-zero.)
Also, since $|x[\alpha]w[\alpha]|$ is known, the STFT measurements
$Y_{r_2,m}(x,w)/|x[\alpha]w[\alpha]|$ give values of the Fourier intensity function $A_{z_{2\alpha}}$ of $z_{2\alpha}$. By~\cite[Corollary IV.3]{bendory2019blind}, the vector~$z_{2\alpha}$ can be determined from $4 \alpha -1$ phaseless measurements.
It follows that for $0\leq \ell \leq \alpha -1$, $x'[2\alpha -\ell ]w'[\ell] =
e^{ 2\iota \theta_\alpha} x[2\alpha -\ell]w[\ell]$. Since we have assumed that
$w'[\ell] = w[\ell] =1$ for $0 \leq \ell \alpha -1$, we deduce
that $x'[\alpha + \ell] = e^{2\iota \theta_\alpha}x[\alpha +\ell]$ for $0 \leq \ell \alpha -1$.

We can now continue by recursion, using $4\alpha -1$
phaseless STFT measurements at each step, to determine that
$y_{j \alpha}(x',w') = e^{\iota j \theta_\alpha}y_{j\alpha}$
%\tb{seems like a typo}
for $j = 3,\ldots , \lceil (N-W-2\alpha)/\alpha \rceil$. This in turn
implies that $x'[n] = e^{\iota \lceil n/\alpha \rceil \theta_\alpha}x[n]$.
However, since our indexing is taken modulo $N$, $x[-n] = x[N-n]$
so that
$e^{\iota \lceil -n/\alpha \rceil \theta_\alpha} = e^{\iota \lceil N-n/\alpha \rceil
  \theta_{\alpha}}$.
Recalling that $N = R \alpha$ we see that this condition is equivalent
to the condition that
$R \theta_\alpha \equiv 0 \bmod 2 \pi$; i.e., $e^{\iota \theta_\alpha}$ is an $R$-th root of unity. Hence, $(x',w')$ is equivalent to $(x,w)$ under the action of
the ambiguity group $G$, as desired.

\section{Numerical experiments} \label{sec:numerical_experiments}
We conducted numerical experiments to examine the bound of Theorem~\ref{thm.known_window}. To recover the signal from samples of its phaseless STFT measurements, we used the relaxed-reflect-reflect (RRR) algorithm, whose $(t+1)$st iteration reads 
\begin{equation}
y^{t+1} = y^{t} + \beta (P_1(2P_2(y^{t})-y^{t})-P_2(y^{t})), 
\end{equation}
where $P_1$ and $P_2$ are projection operators, and $\beta$ is a parameter; we set $\beta=1/2.$
RRR is a general computational framework for constraint satisfaction problems, such as phase retrieval, graph coloring, sudoku, and protein folding~\cite{elser2018benchmark,elser2007searching}. 
In our setting, the algorithm aims to estimate the full $N^2$ STFT entries  (with phases)  from a random subset of its magnitudes. The full STFT
uniquely determines the corresponding signal. 
In particular, in our setting, the first projection, $P_1$, is the orthogonal projector onto
the subspace of matrices which are the STFT of some signal. 
Namely, 
\begin{equation}
P_1 = AA^\dagger,
\end{equation}
where $A$ is the STFT operator as a matrix, and $A^\dagger$ is its pseudo-inverse. 
The second projection,~$P_2$, uses the measured data and is acting by
\begin{equation}
	(P_2z)[i] = \begin{cases}
		\text{sign}(z[i])|y[i]|,& \quad i\in M, \\
		z[i],& \quad i\notin M,
	\end{cases}
\end{equation} 
where $M$ denotes the set of STFT entries for which the magnitudes are known (the measurements), $|y[i]|$ is the $i$th STFT magnitude, and $\text{sign}(z[i]):=\frac{z[i]}{|z[i]|}$.

We use RRR since it is guaranteed to halt only when both constraints are satisfied~\cite[Corollary 4]{levin2019note}. Therefore, we expect (although not guaranteed) to find  a point whose phaseless STFT matches the measurements after enough RRR iterations. The number of iterations required to find such a point provides a measure of hardness~\cite{elser2018benchmark}. In our experiments, we stopped the algorithm when the ratio $||y^{t+1}-y^{t}||/||y^{t}||$  dropped below $10^{-8}$, or after a maximum of $10^4$ iterations. 
We did not conduct experiments for the blind case (Theorem~\ref{thm.unknown_window}) since, as far as know, there is no algorithm that is guaranteed to find a feasible point. 

In our experiments, we set $N=11$ and collected $KN$ STFT magnitudes; the entries were chosen uniformly at random for $K=2,4,6,8.$
The entries of the real underlying signal were drawn from a Gaussian distribution with mean zero and variance 1. 
Note that since the signal is real, the number of parameters to be recovered is $N$, and not $2N$ as in Theorem~\ref{thm.known_window}.
The entries of the window were drawn from the same distribution. 
For each $K$, we conducted 100 trials for each pair of $(L,W)$, where $L=1,\ldots,6$ and $W=1,\ldots,11.$ We  declared a successful trial if the relative error between the estimated signal and the underlying signal (up to a sign)  dropped below $10^{-4}$. 

Figure~\ref{fig:numerical_experiments} reports the success rate and the average number of RRR iterations per $K,W,L$. As expected, the success rate increases with $K$. For $K=2$ ($2N$ STFT magnitudes), 
we can see that for $L\leq 5$ and large enough $W$, the RRR usually does not require many iterations, but it does not always find a solution. Nevertheless,  the success rate is not negligible. 
For $K=6$ and $K=8$, the success rate tends to $1$ for $L\leq 5.$ As can be seen, the true solution is found after a small number of iterations, indicating that the problem is rather easy in this regime. 
Overall, these experiments indicate that indeed a signal can be recovered from a subset of its phaseless STFT magnitudes, and in some cases, quite easily.

\begin{figure}
	\begin{subfigure}[ht]{0.45\columnwidth}
		\centering
		\includegraphics[width=.8\columnwidth]{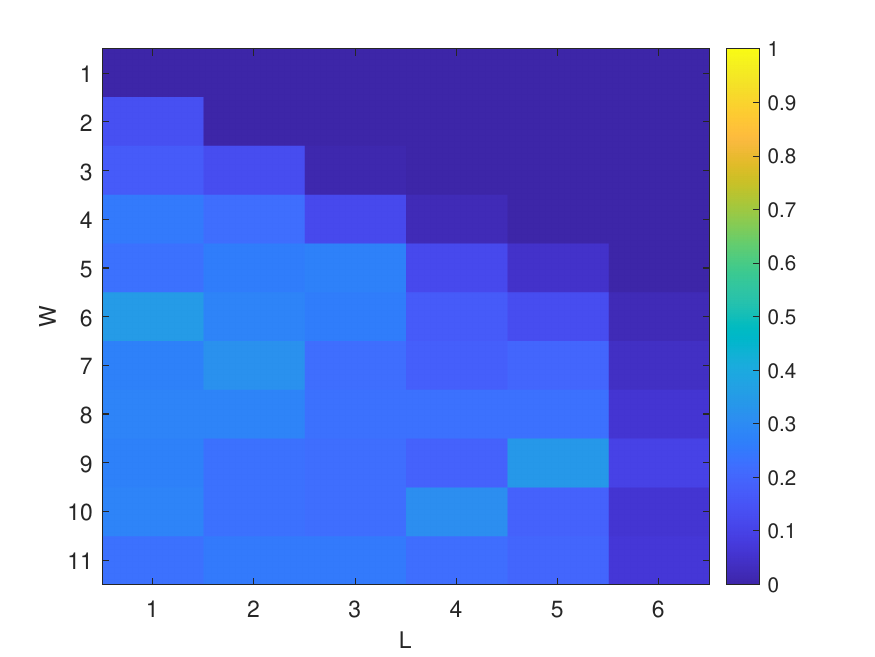}
		\caption{Success rate for $K=2$}
	\end{subfigure}
	\hfill
	\begin{subfigure}[ht]{0.45\columnwidth}
		\centering
		\includegraphics[width=.8\columnwidth]{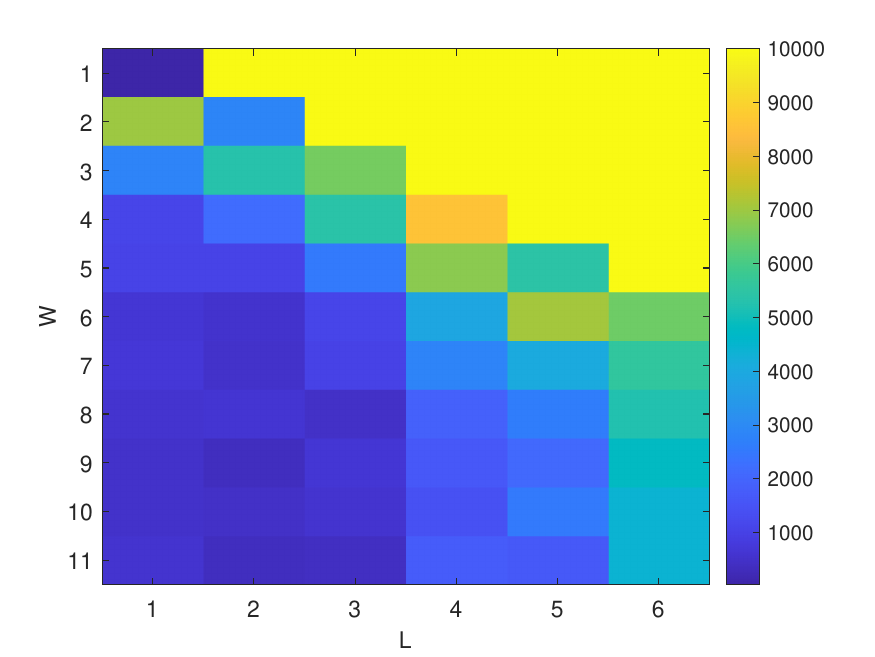}
		\caption{Number of iterations for $K=2$}
	\end{subfigure}

	\begin{subfigure}[ht]{0.45\columnwidth}
	\centering
	\includegraphics[width=.8\columnwidth]{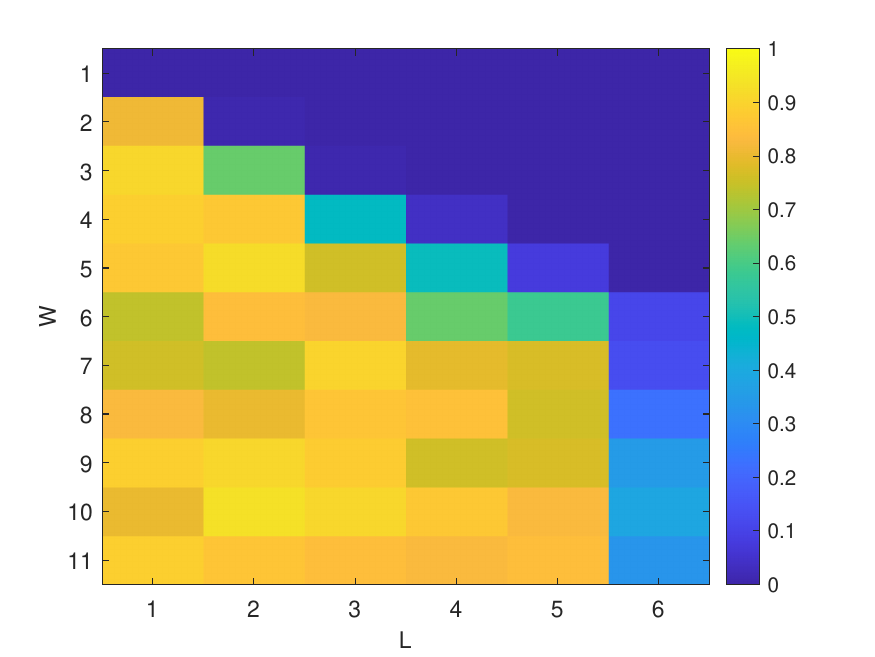}
	\caption{Success rate for $K=4$}
\end{subfigure}
\hfill
\begin{subfigure}[ht]{0.45\columnwidth}
	\centering
	\includegraphics[width=.8\columnwidth]{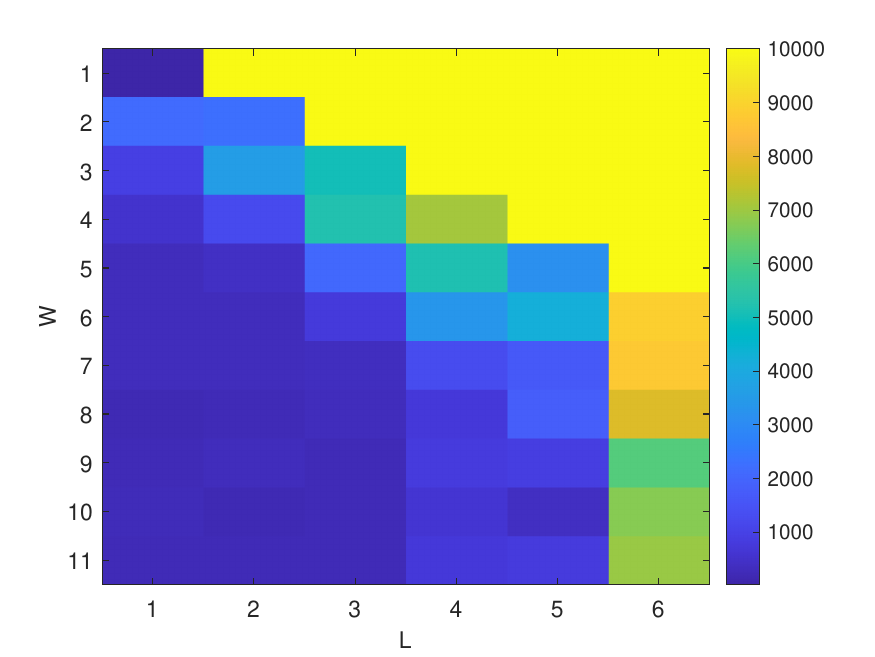}
	\caption{Number of iterations for $K=4$}
\end{subfigure}

	\begin{subfigure}[ht]{0.45\columnwidth}
	\centering
	\includegraphics[width=.8\columnwidth]{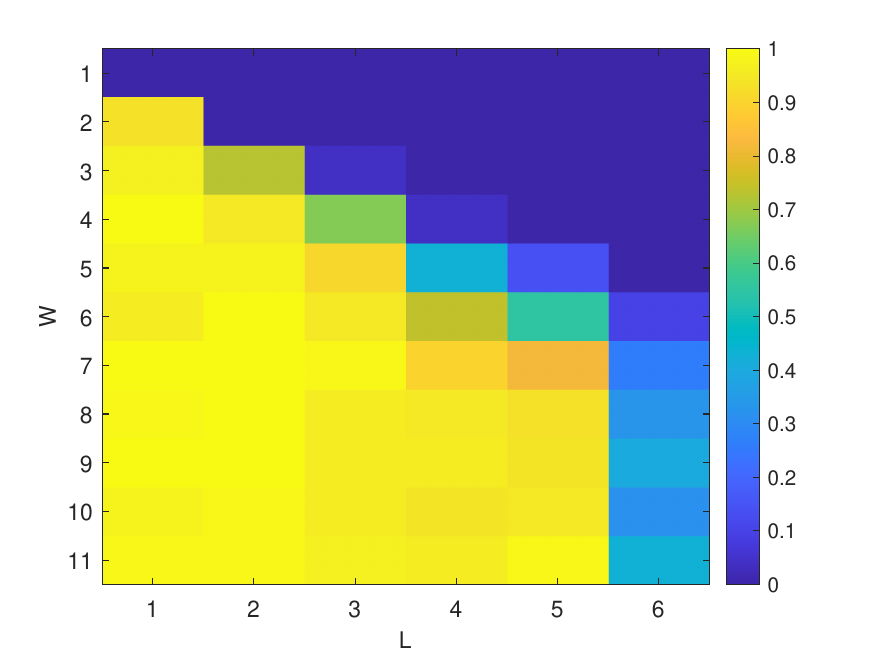}
	\caption{Success rate for $K=6$}
\end{subfigure}
\hfill
\begin{subfigure}[ht]{0.45\columnwidth}
	\centering
	\includegraphics[width=.8\columnwidth]{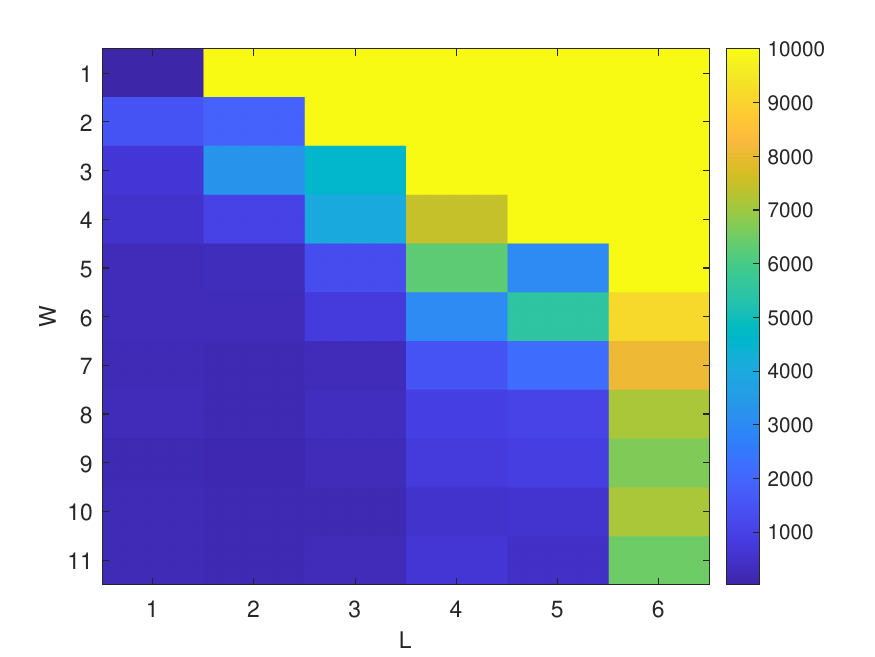}
	\caption{Number of iterations for $K=6$}
\end{subfigure}	\begin{subfigure}[ht]{0.45\columnwidth}
\centering
\includegraphics[width=.8\columnwidth]{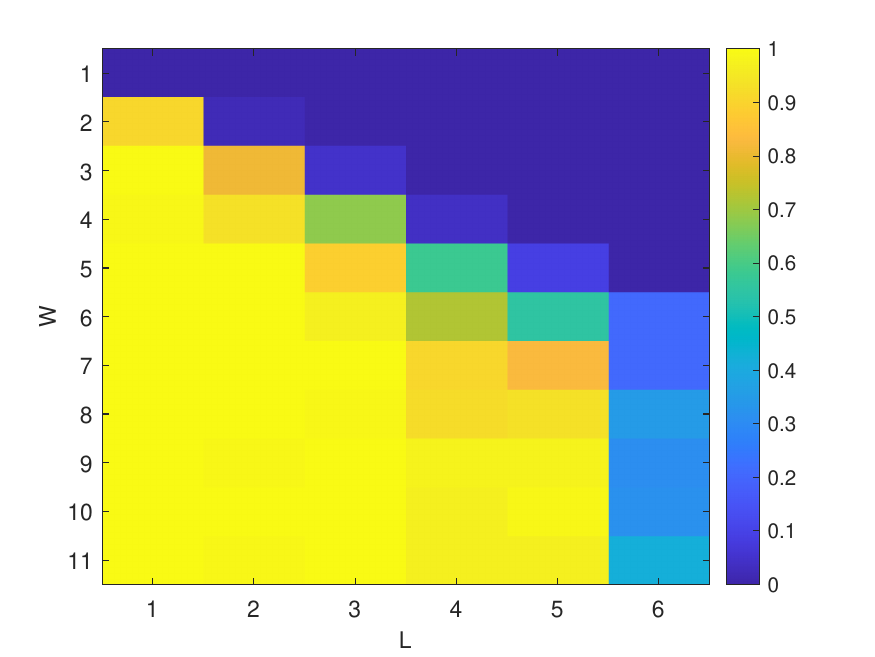}
\caption{Success rate for $K=8$}
\end{subfigure}
\hfill
\begin{subfigure}[ht]{0.45\columnwidth}
\centering
\includegraphics[width=.8\columnwidth]{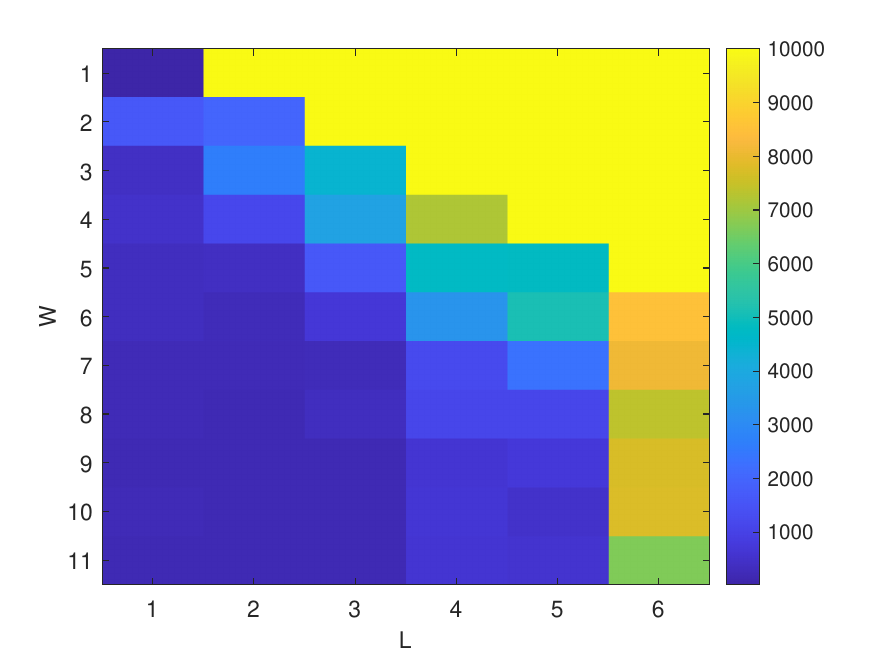}
\caption{Number of iterations for $K=8$}
\end{subfigure}

	\caption{\label{fig:numerical_experiments} The success rate (left column) and average number of iterations (right column) for recovering a signal from its $NK$ phaseless STFT measurements for $K=2,4,6,8$.}
\end{figure}

\section{Orbit frame phase retrieval}
The periodic STFT phase retrieval problem leads to a natural mathematical generalization
which we refer to as {\em phase retrieval for orbit frames}. 
Let $H$ be a compact group acting on~$\C^N$. The {\em orbit} of a possibly unknown
\emph{generating kernel}
$u \in \C^N$ is the set $\{hu| h \in H\}$. An {\em orbit frame} is a matrix  $A\in 
\C^{M \times N}$ $(M\geq N)$  of rank $N$ whose rows are samples of the vectors in $hu$.
The phase retrieval problem for an orbit frame is determining whether
a vector~$x$ can be recovered, up to symmetries, from the phaseless measurements
$|Ax| \in \R_{\geq 0}^M$. 

The definition of orbit frames is broad, and our main focus for future work is the case
where the group $H$ is of the form $G \times \T$, where
$\T$ is  subgroup of
$S^1$ acting on $\C^N$  with weights $(0,1 \ldots , N-1)$, and $G$ is a finite group. In this model, our phaseless
frame measurements on a vector $x$ are samples of the Fourier intensity functions
$|{\widehat {D_1x}}(\omega)|^2, \ldots , |{\widehat{D_rx}}(\omega)|$, where $D_1, \ldots D_r$ 
are diagonal matrices
obtained from the action of the group $G$ on the kernel vector $u$, and $\widehat{D_1x}, \ldots \widehat{D_rx}$ are the Fourier transforms of ${D_1x}, \ldots {D_rx}$. In particular, the periodic STFT model can be thought of as a special case, where $\T = \Z_N$ is the
group of $N$-th roots of unity, $H = \Z_N$ is the group of cyclic translations
and the kernel $u=w$ has support length $W$. (When the kernel $u$ is arbitrary,
this is a Gabor frame; perfect phase retrieval for full Gabor frames was studied
in \cite{bojarovska2016phase}.) The
diagonal matrices $D_1, \ldots D_r$ are  $\diag(w), \diag(T_{L}w),
\ldots \diag(T_{L(R-1)}w)$,
where $T_L$ is the translation operator shifting the entries of $w$ by $L$ entries. 
The phaseless periodic STFT measurements are obtained by sampling the functions $|{\widehat{D_jx}}(\omega)|^2$
at the $N$-th roots of unity. 

The orbit frame phase retrieval problem has been previously studied by a number of authors~\cite{pfander2019robust, bojarovska2016phase,li2019phase,cheng2019twisted} with the main focus
being on constructing large frames, typically of size $M=O(N^2)$, which admit perfect reconstruction from phaseless measurements. 
As in this paper,  we wish to  construct smaller
frames, of size $O(N)$,  for which generic vectors can be recovered from phaseless measurements.
%Specifically, a natural problem for further work is to determine information-theoretic bounds both on
%the number of structured and random phaseless measurements required for
%generic phase retrieval in orbit frames.
Although this problem is mathematically motivated,
understanding the information-theoretic limits of the
general model has the potential to inspire physicists and engineers to develop new measurement techniques.

\section*{Acknowledgment}
This research is support by the BSF grant no. 2020159.
T.B. is also supported in part by the NSF-BSF grant no. 2019752, and the ISF grant no. 1924/21 and D.E. was also supported by NSF-DMS 1906725.

\bibliographystyle{plain}

%\bibliography{ref}

%\bibliographystyle{alpha}

\appendix
    \section{Proof of Proposition~\ref{prop.linearhay}} \label{sec:proof_prop_linhay}
      Let $Z \subset \C^W \times \C^W$ be the linear subspace defined by
      equations~\eqref{eq.linear_relations}. The subspace $Z$ is invariant
      under the action $S^1$, which acts by simultaneous rotation of each
      vector. 
      Let $H$ be the quotient by the $S^1$ action 
      of the open set in $Z$ corresponding to pairs $(z_0, z_\alpha)$
      with $z_0[0], z_0[W-1], z_\alpha[0], z_\alpha[W-1]$
      %\TODO{$z_1$?}
      all non-zero. This implies that the roots $\hat{z_0}(\omega)$
      and $\hat{z_\alpha}(\omega)$ are all non-zero.
      
      Consider the incidence subvariety $I_H \subset H \times H$ consisting of
      pairs of equivalence classes $\left((z_0,z_\alpha), (z'_0, z'_\alpha)\right)$, where $z_\ell$ and $z'_\ell$
      have the same Fourier intensity function. 
      Consider the projection to the first factor $I_H \to H$. Observe 
      that for a given  pair $(z_0,z_\alpha)$ with both vectors non-zero
      there are at most $(2^{W-1})^2$ pairs
      of the form $((z_0,z_\alpha), (z'_0,z'_\alpha))$ in $H$. The reason is as follows.
      We know that for a given vector $z_0$ there are (at most) $2^{W-1}$ vectors
      $z'_{0,j}$ such than any vector $z'_0$ with $A_{z'_0} = A_z$
      must be of the form $z'_0 = e^{\iota \theta_0} z'_{0,j}$ for some $j$.
      Likewise, there are (at most) $2^{W-1}$ vectors $z'_{\alpha,k}$
      such that any vector $z'_\alpha$ with $A_{z'_{\alpha}} = A_{z_\alpha}$
      must be of the form $z'_\alpha = e^{\iota \theta_\alpha} y_{\alpha, k}$ for some
      $k$.       
      However, if we require that the pair $(z'_0, z'_\alpha)$ lies
      in $Z$, then for a given choice of angle $\theta_0$ and vector
      $e^{\iota \theta_0} z'_{0,j}$, there can be at most one angle $\theta_{\alpha}$
      such that the pair  $(e^{\iota_{\theta_0}} z'_{0,j}, e^{\iota \theta_\alpha}
      z'_{1,k})$ satisfies the linear equations~\eqref{eq.linear_relations}.
      
      Note that $I_H$ contains the diagonal $\Delta_H = \{(z_0,z_\alpha),
      (z_0,z_\alpha)| (z_0, z_\alpha) \in H\}$. The above discussion
      shows that $I_H$ has at most $(2^{W-1})^2$ possible components that
      can surject onto $H$. We index the possible components
      as $(I_H)_{j,k}$ with $0 \leq j, k\leq 2^{W-1} -1$
      with $(I_H)_{0,0}$ corresponding to the diagonal.
      
       We will
        show that none of the components can have image all of $H$ by
        explicitly constructing pairs $(z_0,z_\alpha)$ such that
        for every component  $(I_H)_{j,k}$ with $(j,k) \neq (0,0)$
        one of our pairs is not in $(I_H)_{j,k}$, 
        the image of that component. 
        For our first pair, we take $z_0 = (1, 1, \ldots , 1)$ (vector
        of all ones) and $$z_\alpha = (w[\alpha]/w[0], w[\alpha +1]/w[1],
        \ldots w[W-1]/w[W-1- \alpha],a_{W-\alpha}, \ldots, a_{W-1}),$$
        where the $a_k$'s are chosen generically.
        For generic choice of vector $w$
        and $a_{W-\alpha}, \ldots , a_{W-1}$, there will be exactly
        $2^{W-1}$ distinct vectors, up to a global phase, with the same
        Fourier intensity function as $z_\alpha$. On the other hand, $z_0$
        has been chosen so that the roots of its Fourier transform all
        lie on the unit circle, so any vector with same Fourier
        intensity function as~$z_0$ is obtained from~$z_0$ by
        multiplying by a global phase. The choice of $z_\alpha$ implies
        that the only pair in the fiber of the map $I_H \to H$ lying
        over $(z_0,z_\alpha)$ is $(z_0,z_\alpha)$. (Recall that we have
        quotiented out  by a global phase ambiguity in our definition
        of $H$.) This implies that any component $(I_H)_{j,k}$ whose
        image contains $(z_0,z_\alpha)$ must necessarily be of the form
        $(I_H)_{j,0}$ for some $j$, possibly non-zero. Here we use the natural notation that a component $(I_H)_{j,0}$ consists of pairs of the form
        $\left((z_0,z_\alpha), (z'_0,z_\alpha)\right)$.
        %\TODO{The last parenthesis should be part of the text.}

        For our second vector we take $z_\alpha= (1,\ldots, 1)$ (all ones)
        and $$z_0= (w[0]/w[\alpha], w[1]/w[\alpha+1], \ldots ,
        w[W-1-\alpha]/w[W-1], b_{W-\alpha}, \ldots, b_{W-1})$$ where  the $b_k$'s are chosen generically.
        The same reasoning as before implies that the only possible components
        of $I_H$ containing the pair $(z_0, z_\alpha)$ must necessarily be of
        the form $(I_H)_{0,k}$ for some $k$, possibly non-zero.

        Putting this together, we see that the only
        component of $I_H$ that contains both of these test vectors
        is $(I_H)_{0,0}$. Therefore, no other component has image
        all of $H$. Hence, for a generic vector $(z_0,z_\alpha) \in H$, 
        the only pair $(z'_0,z'_\alpha) \in H$ with the same Fourier intensity functions as $(z_0,z_\alpha)$ is $(z_0,z_\alpha)$. This concludes the proof of Proposition
        \ref{prop.linearhay}.
        
        \section{Proof of Proposition~\ref{prop.hardhay}} \label{sec:proof_prop_hardhay}
        The proof of Proposition \ref{prop.hardhay} is similar to the proof
        of Proposition~\ref{prop.linearhay} but more intricate.  Again, let $Z \subset (\C^W)^3$ be the closure of the image of $\C^N \times \C^W$
        under the map
        $$(x,w) \mapsto (z_0(x,w), z_\alpha(x,w), z_{-\alpha}(x,w)).$$ Any triple in $Z$ satisfies the equations~\eqref{eq_blind_triple}.
The group $S^1 \times S^1$ acts on $Z$ with the following action:
        $$\left(e^{\iota \theta_0}, e^{\iota \theta_\alpha}\right) \cdot (z_0, z_\alpha, z_{-\alpha})
= (e^{\iota \theta_0} z_0, e^{\iota \theta_\alpha} z_\alpha, e^{\iota -\theta_\alpha}z_{-\alpha}).$$
Let $H$ be the quotient by $S^1 \times S^1$ of the open set in $Z$ of triples for which $z_0[0], z_0[W-1], z_\alpha[0], z_\alpha[W-1], z_{-\alpha}[0], z_{-\alpha}[W-1]$ are all non-zero and at least one product $z_0[n]z[n+\alpha]$ is non-zero.

        Let $I_H \subset H \times H$ denote the real algebraic subset of
        pairs $\left((z_0,z_\alpha, z_{-\alpha}), (z'_0, z'_\alpha,z'_{-\alpha})\right)$
        such that $A_{z_j} =A_{z'_j}$ for $j = 0, \alpha, -\alpha$. The same argument
        used in the proof of Proposition~\ref{prop.linearhay} shows that
        the polynomial constraint given by \eqref{eq_blind_triple} implies
        that for any triple $(z_0, z_\alpha, z_{-\alpha}) \in H$ there are at most
        $(2^{W-1})^3$ possible pairs of triples
        $\left( (z_0,z_\alpha, z_{-\alpha}), (z'_0, z'_\alpha, z'_{-\alpha})\right)
        \in I_H$. Thus, $I_H$ has at most $(2^{W-1})^3$ components which can dominate
        $H$. We index them by  $(I_H)_{j,k, \ell}$ with $j,k,\ell \in
            [0, 2^{W-1}]$ and the component $(I_H)_{0,0,0}$ is the diagonal.
            
          Again, we will show that the only component of $I_H$ that
          can surject onto $H$ is $(I_H)_{0,0,0}$. 
 
Consider the triple $(z_{0},z_{\alpha},z_{-\alpha})$, where 
\begin{align*}
z_{-\alpha}[i]=\begin{cases}
	1/4& \text{if}\quad  i=0,W-1-\alpha,\\
	1& \text{if} \quad i=W-1,\\
	0& \text{else},
\end{cases}
\end{align*}

\begin{align*}
z_{\alpha}[i]=\begin{cases}
	4& \text{if} \quad i=0,\alpha,W-1-\alpha,W-1,\\
	0& \text{else}.
\end{cases}\end{align*}
\begin{align*}
z_{0}[i]=\begin{cases}
	1& \text{if} \quad i=0,\alpha,W-1-\alpha,W-1,\\
	0& \text{else}.
\end{cases}\end{align*}

This particular triple can be seen to be in the image of the map $\Phi$ by setting $w[0]=w[\alpha]=w[W-1-\alpha]=w[W-1]=1$ and $w[j]=0$ otherwise, and choosing
the values of $x[j]$ accordingly.

The roots of the polynomials $\hat{z}_0(\omega), \hat{z}_\alpha(\omega)$
both lie on the unit circle, while
the roots of $\hat{z}_{-\alpha}(\omega) = 1/4 + 1/4\omega^{W-1-\alpha} + \omega^{W-1}$ all lie strictly inside the unit circle. 
This can be deduced by invoking Cauchy's theorem: the roots of
$1/4 + 1/4z^{W-1-\alpha} + z^{W-1}$ lie strictly inside the unit circle since  the unique positive root
of the polynomial $g(z) = z^{W-1} - 1/4z^{W-1-\alpha} -1/4$ is between $0$ and $1$
since $g(0) < 0$ and $g(1) > 0$. 
%The reason for this is that the unique positive root
%of the polynomial $g(z) = z^{W-1} - 1/4z^{W-1-\alpha} -1/4$ is between $0$ and $1$
%since $g(0) < 0$ and $g(1) > 0$. Thus by Cauchz's theorem the roots of
%$1/4 + 1/4z^{W-1-\alpha} + z^{W-1}$ lie strictly inside the unit circle.

Now, if $(z'_0, z'_{\alpha}, z'_{-\alpha})$ is a triple such that
$A_{z'_{\ell}} = A_{z_{\ell}}$ for $\ell \in \{0,\alpha, -\alpha\}$, then $z'_0, z'_\alpha$ are obtained
from $z_0, z_\alpha$ by multiplication by a global phase, because all of the roots
of $\hat{z}_0(\omega) , \hat{z}_\alpha(\omega) $ lie on the unit circle. On the other hand, since all of the of the roots of $\hat{z}_{-\alpha}(\omega)$ are distinct and none 
lie on the unit circle, there are, up to a global phase, $2^{W-1}$ vectors
$z'_{-\alpha}$. We will show that the triple $(z'_0, z'_\alpha, z'_{-\alpha})$ is in $H$ if and only if
$z'_{-\alpha}$ is obtained from $z_{-\alpha}$ by multiplication by a global
phase. To see this, note
that if $(\beta_1, \ldots , \beta_{W-1})$ are the roots of
the polynomial $1/4 + 1/4\omega^{W-1-\alpha} + \omega^{W-1}$, then
$$\hat{z}'_{-\alpha}(\omega) = \prod_{n \in I}|\beta_n| (\omega-\beta_i/|\beta_i|^2) \prod_{n \notin I} (\omega- \beta_n),$$ for some subset $I \subset [1,W-1]$. Since $|\beta_n| < 1$ because $(\beta_1, \ldots , \beta_{W-1})$ lie inside the unit circle
%\TODO{why?},
the constant term of $\hat{z}'_{-\alpha}(\omega)$
will be strictly greater than $1/4$, 
making it impossible for triple $(z'_0,z'_{\alpha}, z'_{-\alpha})$ to satisfy the constraints of~\eqref{eq_blind_triple}. This implies that
any component $(I_H)_{j,k, \ell}$ that contains the triple
$(z_0, z_\alpha, z_{-\alpha})$  in its image must be of the form $(I_H)_{j,k,0}$ for
some $j,k$. Hence, any component of $I_H$ which dominates $H$ must
be of the form $(I_H)_{j,k,0}$.

%This \TODO{which?} choice of vector shows that any component of $I_H$ that surjects onto
%$H$ cannot be of the form $I_{j,k,\ell}$ with $\ell \neq 0$. In other words,
%any possible component of $I_H$ which dominates
%$H$ must be of the form $(I_H)_{j,k,0}$. \TODO{This seems to be somewhat out of context?}

Now consider the triple $(z_0, z_{\alpha}, z_{-\alpha})$ with
$z_{-\alpha} = z_{0} = (1,\ldots , 1)$ (all ones), and
$z_\alpha = (c,1, \ldots , 1)$ with $c > W-1$.
The polynomial $z_\alpha(\omega)  = c+ \omega + \omega^2 \ldots + \omega^{W-1}$ has all roots outside
the unit circle, since $|\omega+ \ldots +  \omega^{W-1}| < W-1 < c$ for any $\omega$ inside the unit circle. If $A_{z'_\alpha}=A_{z_\alpha}$ and $(\beta_1, \ldots, \beta_{W-1})$
are the roots of $\hat{z_\alpha}(\omega)$, then
$$\hat{z'}_\alpha(\omega) = \prod_{\ell \in I}|\beta_{\ell}|(\omega- \beta_\ell/|\beta_\ell|^2) \prod_{\ell \notin I} (\omega- \beta_\ell),$$
for some subset $I \subset [1, W-1]$.
In particular, it follows that $|z'_\alpha[W-1]| > 1$ since $|\beta_{\ell}| > 1$
for all $\ell$ unless $|I| = \emptyset$.
On the other hand, all roots of $\hat{z}_0(\omega)$ and $\hat{z}_{-\alpha}(\omega)$ lie on the unit
circle, so if $A_{z'_0} = A_{z_0}$ and $A_{z'_{-\alpha}} = A_{z_{-\alpha}}$ then
$z'_0, z'_{\alpha}$ are obtained from $z_0, z_{\alpha}$ by a global phase change
and the magnitude of the entries are unchanged. Hence, the triple
$(z'_0, z'_\alpha, z'_{-\alpha})$ cannot satisfy equations \eqref{eq_blind_triple} unless $(z'_0, z'_\alpha, z'_{-\alpha})$ is obtained from $(z_0, z_\alpha, z_{-\alpha})$ by a global phase. Thus, the only possible components of
$I_H$ which dominate $H$ are of the form $(I_H)_{j,0,0}$.

To show that
a component of the form $(I_H)_{j,0,0}$ does not have image all of $H$ unless $j=0$, 
it suffices to show that there exists a triple $(z_0, z_{-\alpha}, z_{\alpha})$
in $H$ such that if $(z'_0, z_{-\alpha}, z_{\alpha})\in H$ and
$A_{z'_0} = A_{z_0}$, then $z'_0$ is obtained from $z_0$ by a global phase.
Note that any vector $z_0$ can be part of a triple in $H$, since for
any given $z_0$, the system of equations 
\begin{equation*}
z_{-\alpha}[\ell]
z_{\alpha}[\alpha + \ell] = z_0[\ell]z_0[\ell + \alpha], \quad \ell = 0, \ldots , W-1-\alpha,
\end{equation*}
%\{z_{-\alpha}[\ell]
%z_{\alpha}[\alpha + \ell] = z_0[\ell]z_0[\ell + \alpha]\}_{\ell = 0, \ldots , W-1-\alpha}$$
has positive dimensional solution space. We claim that we can choose a vector
$z_0$ such that if $z'_0$ does not differ from $z_0$ by a global phase, then
$|z_0[W-1] z_0[W-1-\alpha]| \neq |z_0'[W-1]z_0'[W-1-\alpha]|$. This follows
from a similar argument used in the proof \cite[Theorem 3.1]{MR3842644}.
If
$\beta_1, \ldots \beta_{W-1}$ are the roots of the
$\hat{z}_0(\omega)$,  then $|z_0[W-1] z_0[W-1-\alpha]|= |z'_0[W-1] z'_0[W-1-\alpha]|$
for some $z'_0$, only if $|S_\alpha(\beta_1, \ldots , \beta_{W-1})|
=\prod_{\ell \in I} |\beta_i| S_\alpha(\beta'_1, \ldots , \beta'_{W-1})$ where
$\beta'_i \in \{\beta_i, \beta_i/|\beta_i|^2\}$ and
$I \subset [1,W-1]$ is the subset where $\beta'_i = \beta_i/|\beta_i|^2$.
For general choice
of $(\beta_1,\ldots, \beta_{W-1})$,  these equations are not satisfied
unless $\beta_i' = \beta_i$ for all $i$. Hence,
$(I_H)_{j,0,0}$ does not surject onto $H$ unless $j = 0$. This concludes the proof of Proposition \ref{prop.hardhay}.
\end{document}